\documentclass[10pt,journal,onecolumn]{IEEEtran}

\usepackage{cite}
\usepackage{mathrsfs}
\usepackage{amsfonts}
\usepackage{amsmath,amssymb,amsthm}
\usepackage{graphicx}

\newtheorem{theorem}{Theorem}

\newtheorem{claim}{Claim}

\newtheorem{corollary}[theorem]{Corollary}
\newtheorem{remark}{Remark}

\newtheorem{definition}{Definition}

\begin{document}

\title{On the Message Dimensions of \\Vector Linearly Solvable Networks}
\author{Niladri Das and Brijesh Kumar Rai}

\maketitle

\begin{abstract}
It is known that there exists a network which does not have a scalar linear solution over any finite field but has a vector linear solution when message dimension is $2$ \cite{medard}. It is not known whether this result can be generalized for an arbitrary message dimension. In this paper, we show that there exists a network which admits an $m$ dimensional vector linear solution, where $m$ is a positive integer greater than or equal to $2$, but does not have a vector linear solution over any finite field when the message dimension is less than $m$. 
\end{abstract}

\section{Introduction}

Network coding - which emerged as an improvement over routing - has been an intense area of research since its inception in the year 2000\cite{alshwede}. It was shown in \cite{alshwede} that for the multicast networks, the min-cut bound can be achieved using network coding. Moreover, scalar linear network coding over a sufficiently large finite field was shown to be sufficient to achieve the capacity of a multicast network\cite{li}. However, for the non-multicast networks, M\'{e}dard \textit{et al.}\cite{medard} presented an instance of a network coding problem which does not admit a scalar linear solution over any finite field but has a $2$ dimensional vector linear solution over every finite field. This network is known as the M-network in the literature. It is a natural question whether, for any positive integer $m\geq 2$, there exists a network which admits an $m$ dimensional vector linear solution but has no vector linear solution over any finite field when message dimension is less than $m$. In this paper, we show that indeed there exists such a network. 

The prior works related to the problem addressed in this paper are as follows. \cite{insufficient,Sun} studied the effect of message dimension on vector linear solvability. In \cite{insufficient}, a network was presented which has an $m$ dimensional vector linear solution over $\mathbb{F}_2$ for an arbitrary positive $m$, but does not have an $m$ dimensional vector linear solution over $\mathbb{F}_q$ for odd $q$ and any $m$. In \cite{Sun}, a multicast network was presented which has a vector linear solution over $\mathbb{F}_2$ for message dimension $4$ but not for message dimension $5$. In \cite{dougherty}, it was shown that the notion of scalar linear solvability of networks can be captured by matroids. Specifically, it was shown that a network is scalar linearly solvable only if it is a matroidal network associated with a representable matroid over a finite field. The converse of this result was proved in \cite{kim}. \cite{sundar} generalized the results of \cite{dougherty,kim} for the vector linearly solvable networks and showed that the existence of an $m$ dimensional vector linear network code solution implies the existence of a discrete polymatroid with certain properties and vice versa. To show that for any positive integer $m\geq 2$, there exists a network which has an $m$ dimensional vector linear solution but has no vector linear solution for a message dimension less than $m$, either one has to give construction of such a network, or equivalently, one can construct a discrete polymatroid such that it is representable if the rank of its every element is allowed to be less than or equal to $m$, but not representable if the rank of its every element is strictly lesser than $m$ \cite{sundar}. To the best of our knowledge, neither such a network has been presented in the literature nor it has been shown that for every integer $m\geq 2$, there do not exist such networks; likewise in the area of matroid theory,  no discrete polymatroid having the above mentioned property has been reported nor it has been shown that such a discrete polymatroid does not exist.

The organization of the rest of the paper is as follows. In Section~\ref{sec1}, we present the formal definitions related to network coding. In Section~\ref{sec2}, we present the main result of the paper. Section~\ref{sec3} is devoted to the proof of the main theorem of the paper. We conclude the paper in Section~\ref{sec4}.

\section{Preliminaries} \label{sec1}

The networks considered in this paper are directed acyclic networks. A directed acyclic network can be modelled by a directed acyclic graph $G=(V,E)$, where $V$ is the set of nodes and $E\subseteq V\times V$ is the set of edges. For an edge $e = (u,v)$, $u$ is denoted by $tail(e)$ and $v$ is denoted by $head(e)$. For a node $v$, the set of edges $\{(x,v)| x \in V
\}$ is denoted by $In(v)$. A set of nodes, denoted by $S \subset V$, will be called as sources and a set of nodes, denoted by $T \subset V$, will be called as terminals. Without loss of generality (w.l.o.g), we assume that a source does not have any incoming edge and a terminal does not have any outgoing edge. Associated with every source, there is an i.i.d. random process which is uniformly distributed over a finite field $\mathbb{F}_q$. Each source process is independent of all other source processes. The source process associated with a source $s_i \in S$ is indicated by $X_i$, and the message carried over an edge $e$ is indicated by $Y_e$. Each terminal demands a subset of source processes. All the edges in the network are assumed to be unit capacity edges. An $(n,n)$ vector linear network code over a finite field $\mathbb{F}_q$ can be described in terms of the message passing through every edge and the decoding function at every terminal. For an edge $e$, when $tail(e)=s_i\in S$, $Y_e = A_{\{s_i,e\}}X_i$, where $X_i, Y_e\in \mathbb{F}_q^n , \text{and}~A_{\{s_i,e\}} \in \mathbb{F}_q^{n\times n}$. When $tail(e)$ is an intermediate node (a node other than a source or a terminal), $Y_e = \sum_{e^{\prime}\in In(tail(e))} A_{\{e^{\prime},e\}} Y_{e^{\prime}}$, where $Y_e,Y_{e^{\prime}} \in \mathbb{F}_q^n$ and  $A_{\{e^{\prime},e\}}\in \mathbb{F}_q^{n\times n}$. And for a terminal node $t_i$, $Y_{t_i} = \sum_{e^{\prime}\in In(t)} A_{\{e^{\prime},t_i\}}Y_{e^{\prime}}$, where $Y_{e^{\prime}},Y_{t_i}\in \mathbb{F}_q^n, \text{and}~A_{\{e^{\prime},t_i\}}\in \mathbb{F}_q^{n\times n}$.

A $(1,1)$ vector linear code is referred as a scalar linear code. In an $(n,n)$ vector linear code, $n$ is said to be the message dimension. An $(n,n)$ vector linear code is termed as an $n$ dimensional vector linear code. A network is said to have an $(n,n)$ vector linear solution if there exists an $(n,n)$ vector linear code for the network which enables all the terminals in the network to receive $n$ source symbols (in $n$ uses of every edge) from each of their respective demanded sources. If a network has an $n$ dimensional vector linear solution then it is said to be vector linearly solvable (for $n=1$, scalar linearly solvable).

\section{Main Result} \label{sec2}

\begin{figure}[!t]
\centering
\includegraphics[width=0.5\textwidth]{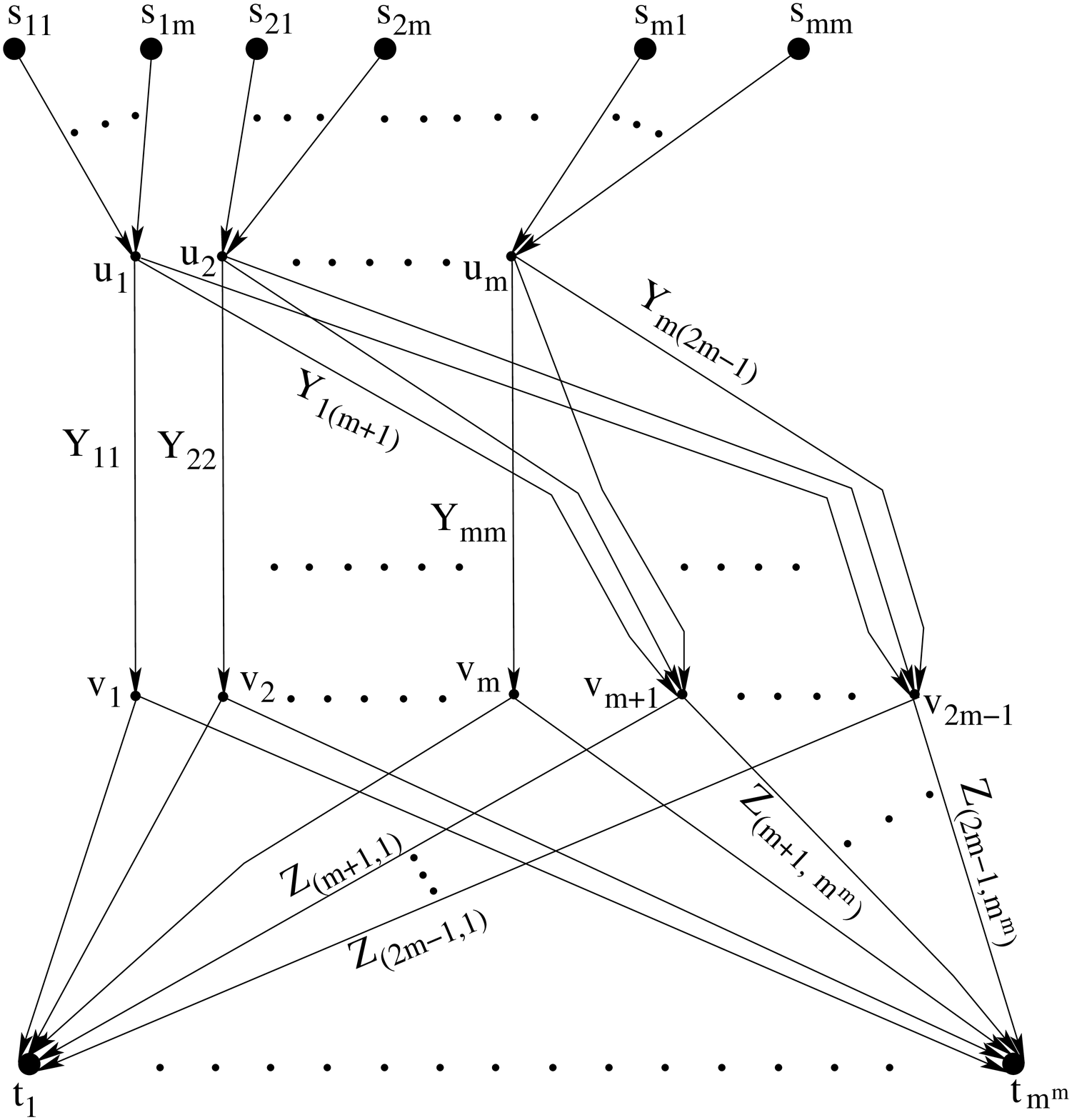}
\caption{A communication network $\mathcal{N}$, which we name as ``generalized M-network'', which is vector linearly solvable only when the message dimension is positive integer multiple of $m$}
\label{general}
\end{figure}

In Fig.~\ref{general}, we present a network $\mathcal{N}$ which has a vector linear solution in $m \geq 2$ message dimension but has no vector linear solution if the message dimension is less than $m$. $\mathcal{N}$ has $m^2$ sources and $m^m$ terminals. The sources are partitioned into $m$ sets. The $j^{\text{th}}$ source in the $i^{\text{th}}$ set is denoted by $s_{ij}$. The source $s_{ij}$ generates the message $X_{ij}$. Below we list the edges in the network:
\begin{enumerate}
\item An edge $(s_{ij}$,$u_i)$ for $1\leq i,j\leq m$
\item An edge $e_{ii} = (u_i,v_i)$ and an edge $e_{ij} = (u_i,v_j)$ for $1\leq i \leq m$ and $m+1\leq j\leq 2m-1$
\item An edge $(v_i$,$t_j)$ for $1\leq i\leq 2m-1$  and \mbox{$1\leq j\leq m^m$}
\end{enumerate}
The message transmitted over the edge $(s_{ij}$,$u_i)$ is denoted by $X_{ij}$. The message transmitted over the edge $e_{ij}$ is denoted by $Y_{ij}$. The message transmitted over the edge $(v_{i},t_{j})$ is denoted by $Z_{ij}$ for $m+1\leq i\leq 2m-1$ and $1\leq j\leq m^m$. Each terminal demands a unique tuple of $m$ source messages where the $i^{th}$ element of the tuple is $X_{ij}$ for any $j$ in the range: $1\leq j\leq m$. W.l.o.g, we assume that $t_1$ demands the source messages: $(X_{11},X_{21},\ldots,X_{m1})$.

We note that for $m=2$, $\mathcal{N}$ is the M-network. Therefore, the presented network $\mathcal{N}$ is a generalisation of the M-network. 

\begin{theorem}\label{thm1}
For an arbitrary finite field $\mathbb{F}_q$, the network $\mathcal{N}$ has a $d$ dimensional vector linear solution over $\mathbb{F}_q$ if and only if $d$ is a multiple of $m$.
\end{theorem}
As a consequence, we have the following corollary.

\begin{corollary}
The network $\mathcal{N}$ is not vector linearly solvable for any message dimension less than $m$, but has an $m$ dimensional vector linear solution.
\end{corollary}

\begin{remark}
It was shown in \cite{dougherty} that the M-network does not have an odd dimensional vector linear solution. This result is a special case of Theorem \ref{thm1} for the case $m=2$. 
\end{remark}

\section{Proof of Theorem \ref{thm1}} \label{sec3}
Our proof of Theorem~\ref{thm1} relies on a result related to discrete polymatroids \cite{sundar}. Therefore to make the proof accessible, in the following, the definitions of discrete polymatroid, representable discrete polymatroid and discrete plolymatroidal network are given \cite{sundar,hibi,poly}. The definitions are borrowed from \cite{sundar}. We first fix some notations. Let $N = \{1,2,\ldots n\}$. The set of all non-negative integers are denoted by $\mathbb{Z}_{\geq 0}$ and the set of positive integers are denoted by $\mathbb{Z}_{> 0}$. The notation $\mathbb{Z}_{\geq 0}^{n}$ indicates the set of all $n$ length vector whose elements are in $\mathbb{Z}_{\geq 0}$. For an $n$ length vector $v$, and $A\subseteq N$, $v(A)$ is the vector having only the components indexed by the elements of $A$, and $|v(A)|$ is the sum of the components of $v(A)$.
 
\begin{definition}
Let $\rho$ be a function that maps $2^N$ into $\mathbb{Z}_{\geq 0}$ and satisfies the following three conditions:\\
(1)\hspace{1em}  $\rho(\emptyset) = 0$.\\
(2)\hspace{1em} If $X\subseteq Y\subseteq N$ then $\rho(X)\leq\rho(Y)$.\\
(3)\hspace{1em} If $X,Y \subseteq N$, then $\rho(X\cup Y) + \rho(X\cap Y) \leq \rho(X) + \rho(Y)$\label{p3}.\\
Let $\mathbb{D}$ be the collection of all elements $v\in \mathbb{Z}_{\geq 0}^{n}$ such that $|v(A)| \leq \rho(A)$ for $\forall A\subseteq N$. Then $\mathbb{D}$ is a discrete polymatroid having rank function $\rho$ and ground set $N$.
\end{definition} 
A discrete polymatroid $\mathbb{D}$ has $\rho_{max}=d$ if $\rho$ follows this additional rule: for $\forall A\subseteq N$, $\rho(A)\leq d|A|$. 

By setting $\rho_{max} = 1$, one can completely describe a matroid with a discrete polymatroid \cite{sundar}. So the discrete polymatroids can be viewed as the generalized version of matroids. Conditions (1)-(3) are known as the \textit{polymatroidal axioms}. It is known that if the entropy function $H(\,)$, in a Shannon-type inequality, is replaced by a function that obeys the polymatroidal axioms then the inequality still remains valid  \cite{dougherty}. We will use this fact in our proof.

\begin{definition}
A discrete polymatroid $\mathbb{D}$ with rank function $\rho$ and ground set $N$, is said to be representable over $\mathbb{F}_q$ if there exist vector subspaces $V_1,V_2,\ldots,V_n$ of a vector space $V$ over $\mathbb{F}_q$ such that $dim(\sum_{i\in X} V_i) = \rho(X)$ for $\forall X\subseteq N$. The set of vector spaces $V_i, i\in N$ is said to form a representation of $\mathbb{D}$. A discrete polymatroid is said to be representable if it is representable for some field.
\end{definition}

\noindent Let $\epsilon_{in}$ be an $n$ length vector whose $i^{\text{th}}$ component is one and all other components are zero. Consider a network in which the set of source messages, the set of nodes and the set of messages carried over by the edges are denoted by $\mu$, $V$ and $\delta$ respectively. For a non-source node $v\in V$, let $I(v)$ be the set of messages carried by the edges in $In(v)$. When $v$ is a source, let $I(v)$ denote the message generated at $v$. Also, let $O(v)$ be the set of messages carried by the edges in $Out(v)$. Let $\mathbb{D}$ be a discrete polymatroid with ground set $N$ and the rank function $\rho$ where $\rho_{max} =d$.

\begin{definition}\label{discrete}
A network is a \textit{discrete polymatroidal network} associated with $\mathbb{D}$ if there exists a function $f:\mu\cup \delta \rightarrow N$ such that the following conditions are satisfied:\\
(1)\hspace{1em} $f$ is one-to-one on $\mu$.\\
(2)\hspace{1em} $\sum_{i\in f(\mu)}d\epsilon_{in}\in\mathbb{D}$.\\
(3)\hspace{1em} $\rho(f(I(x))) = \rho(f(I(x)\cup O(x)))$, $\forall x \in V$. 
\end{definition}
Note that if a network is matroidal (defined in \cite{dougherty}) then it is also discrete polymatroidal for $\rho_{max}=1$. 
The following theorem is reproduction of \textit{Theorem 1} from \cite{sundar}.

\begin{theorem}\label{thm:sundar}
A network has a $k$ dimensional vector linear solution over $\mathbb{F}_q$ if and only if it is discrete polymatroidal with respect to a discrete polymatroid $\mathbb{D}$ representable over $\mathbb{F}_q$ with $\rho_{max} = k$.
\end{theorem}

\textit{Proof of Theorem~\ref{thm1}:}
The proof is similar to that used in \cite{dougherty} to show that the M-network is not matroidal. Say the function $f$ maps the network $\mathcal{N}$ to a discrete polymatroid $\mathbb{D}$ conforming to the rules of mapping presented in Definition~\ref{discrete}. Also let $\mathtt{g} = \rho \circ f$ where $\rho$ is the rank function of $\mathbb{D}$. Assume $\rho_{max} = d$. Our proof depends on the following two sets of inequalities.
\textbf{Set I}: 
\begin{equation}
\mathtt{g}(Y_{11},X_{1j_1}) + \mathtt{g}(Y_{22},X_{2j_2}) + \cdots + \mathtt{g}(Y_{mm},X_{mj_m}) \leq (2m-1)d \; \text{ for }  j_i\in \{1,2,\ldots,m\}, 1\leq i \leq m 
\end{equation}
\textbf{Set II}:
\begin{equation}
\mathtt{g}(Y_{ii},X_{i1}) + \mathtt{g}(Y_{ii},X_{i2}) + \cdots + \mathtt{g}(Y_{ii},X_{im})
\geq  (2m-1)d \; \text{ for } 1\leq i\leq m 
\end{equation}
\begin{claim}\label{cla1}
The inequalities in Set~I hold true.
\end{claim}
\begin{proof}
We give the proof for $j_i = 1$, for $1\leq i \leq m$. The rest of the inequalities can be proved in a similar way. To prove our claim we use the following fact about the entropy function: for any set of random variables $A_j$, $1\leq j\leq k$,
\begin{equation}
H(A_1) + H(A_2) + \cdots + H(A_k) = H(A_1,A_2,\ldots,A_k)
\text{ iff } H(A_i|A_1,A_2,\dots,A_{i-1}) = H(A_i) \text{ for } 2\leq i\leq k
\end{equation}
Now in $\mathcal{N}_1$, \mbox{$H(Y_{22},X_{21}|Y_{11},X_{11}) = H(Y_{22},X_{21})$}.\\
Similarly, for \mbox{$1\leq i,j,l,k\leq m$}, $H(Y_{ii},X_{ij}|\cup_{k\neq i}Y_{kk},\cup_{k\neq i}X_{kl}) = H(Y_{ii},X_{ij})$. Therefore,
\begin{IEEEeqnarray}{r}
H(Y_{11},X_{11})+H(Y_{22},X_{21})+\cdots +H(Y_{mm},X_{m1})= H(Y_{11},X_{11},Y_{22},X_{21},\ldots ,Y_{mm},X_{m1})\label{ent1}
\end{IEEEeqnarray}
Since $g$ obeys polymatroid axioms, it also obeys (\ref{ent1}) if $H$ is replaced by $g$. Thus, 
\begin{IEEEeqnarray*}{l}
\mathtt{g}(Y_{11},X_{11})+\mathtt{g}(Y_{22},X_{21})+\cdots +\mathtt{g}(Y_{mm},X_{m1})\\
= \mathtt{g}(Y_{11},X_{11},Y_{22},X_{21},\ldots ,Y_{mm},X_{m1})\\
\leq \mathtt{g}(Y_{11},X_{11},Y_{22},X_{21},\ldots ,Y_{mm},X_{m1},Z_{(m+1,1)},Z_{(m+2,1)},\ldots ,Z_{(2m-1,1)})\\
= \mathtt{g}(Y_{11},Y_{22},\ldots ,Y_{mm},Z_{(m+1,1)},\ldots ,Z_{(2m-1,1)})\IEEEyesnumber\label{a1}\\
\leq (2m-1)d \IEEEyesnumber \label{a2}
\end{IEEEeqnarray*}

The equation~(\ref{a1}) is true because the terminal $t_1$ computes $(X_{11},X_{21},\ldots,X_{m1})$ from the messages $\{Y_{11},Y_{22}, \ldots ,Y_{mm},\\ Z_{(m+1,1)},Z_{(m+2,1)},\ldots ,Z_{(2m-1,1)}\}$. Equation~(\ref{a2}) is true because each element can have rank maximum of $d$ and there are $(2m-1)$ elements. This concludes the proof of Claim~\ref{cla1}.
\end{proof}

\begin{claim}\label{cla2}
The inequalities in Set~II hold true.
\end{claim}

\begin{proof}
We will give the proof of the inequality for $i=1$. The rest can be proved similarly. First we show that $\mathtt{g}(Y_{ii}) =\mathtt{g}(Y_{ij})= d$ for $1\leq i\leq m, m+1\leq j\leq 2m-1$. Since all source messages are independent,
\begin{IEEEeqnarray*}{l}
m^2d = \mathtt{g}(X_{11},X_{12},\ldots ,X_{mm})\\
\leq  \mathtt{g}(X_{11},\ldots ,X_{mm},Y_{11},\ldots ,Y_{mm},Y_{1(m+1)},\ldots ,Y_{m(2m-1)})\\
=  \mathtt{g}(Y_{11},Y_{22},\ldots ,Y_{mm},Y_{1(m+1)},Y_{1(m+2)},\ldots ,Y_{m(2m-1)})\label{2a}\IEEEyesnumber\IEEEeqnarraynumspace\\
\leq  \mathtt{g}(Y_{11}) + \mathtt{g}(Y_{22}) + \cdots + \mathtt{g}(Y_{mm}) + \mathtt{g}(Y_{1(m+1)}) + \mathtt{g}(Y_{1(m+2)}) + \cdots + \mathtt{g}(Y_{m(2m-1)})\\
\leq  md + m(m-1)d= m^2d
\end{IEEEeqnarray*}

Equality in (\ref{2a}) follows because every symbol is demanded by some terminal. Hence, 
\begin{equation}
\mathtt{g}(Y_{11}) + \mathtt{g}(Y_{22}) + \cdots + \mathtt{g}(Y_{mm}) + \mathtt{g}(Y_{1(m+1)}) + \mathtt{g}(Y_{1(m+2)}) + \cdots + \mathtt{g}(Y_{m(2m-1)}) = m^2d
\end{equation}  
Since, there are $m^2$ terms and each term can take a maximum value of $d$, $\mathtt{g}(Y_{ii}) = \mathtt{g}(Y_{ij}) = d$ for $1\leq i\leq m$ and $m+1\leq j\leq 2m-1$. Now we prove the inequality:
\begin{IEEEeqnarray*}{l}
\mathtt{g}(Y_{11},X_{11}) + \mathtt{g}(Y_{11},X_{12}) + \cdots + \mathtt{g}(Y_{11},X_{1m})\\
\geq \mathtt{g}(Y_{11},X_{11},X_{12}) + \mathtt{g}(Y_{11}) + \cdots +\mathtt{g}(Y_{11},X_{1m})\\
\geq \mathtt{g}(Y_{11},X_{11},X_{12},X_{13}) + 2\mathtt{g}(Y_{11}) + \cdots + \mathtt{g}(Y_{11},X_{1m})\\
\qquad \qquad \textsf{:}\qquad \qquad   \textsf{:} \qquad \qquad  \textsf{:}\\ 
\geq \mathtt{g}(Y_{11},X_{11},X_{12},\ldots ,X_{1m}) + (m-1)\mathtt{g}(Y_{11})\\
= md + (m-1)d = (2m-1)d
\end{IEEEeqnarray*}

Here we have used condition (3) from definition~\ref{p3} repeatedly. This concludes the proof of Claim~\ref{cla2}.
\end{proof}
We prove the theorem by finding a constraint on the rank function using the inequalities in Set I and Set II. We show that $\mathtt{g}(Y_{ii},X_{ij}) = \frac{(2m-1)d}{m} \text{ for } 1\leq i,j\leq m$. We will give the proof only for $\mathtt{g}(Y_{mm},X_{m1}) = \frac{(2m-1)d}{m}$. The rest can be proved similarly. To prove that $\mathtt{g}(Y_{mm},X_{m1}) = \frac{(2m-1)d}{m}$, we consider an inequality from Set I which has $\mathtt{g}(Y_{mm},X_{m1})$ on the left hand side. We then eliminate (one by one) all the rest terms except $\mathtt{g}(Y_{mm},X_{m1})$ from left hand side using other inequalities from the Set I and inequalities from Set II.  Consider the inequality:
\begin{IEEEeqnarray}{l}
\mathtt{g}(Y_{11},X_{11}) {+} \mathtt{g}(Y_{22},X_{21}) {+} \cdots {+} \mathtt{g}(Y_{mm},X_{m1}) \leq  (2m-1)d \nonumber \\
\end{IEEEeqnarray}
Now consider all other inequalities from Set I which differ only at the first term than the above inequality. There are exactly $m-1$ such inequalities. These inequalities are written below: 
\begin{IEEEeqnarray*}{l}
\mathtt{g}(Y_{11},X_{12}) {+} \mathtt{g}(Y_{22},X_{21}) {+} \cdots {+} \mathtt{g}(Y_{mm},X_{m1}) \leq  (2m-1)d\\
\mathtt{g}(Y_{11},X_{13}) {+} \mathtt{g}(Y_{22},X_{21}) {+} \cdots {+} \mathtt{g}(Y_{mm},X_{m1}) \leq  (2m-1)d\\
\qquad \qquad \textsf{:}\qquad \qquad   \textsf{:} \qquad \qquad  \textsf{:}\\ 
\mathtt{g}(Y_{11},X_{1m}) {+} \mathtt{g}(Y_{22},X_{21}) {+} \cdots {+} \mathtt{g}(Y_{mm},X_{m1}) {\leq}  (2m-1)d
\end{IEEEeqnarray*}
Summing up all of the above $m-1$ inequalities and the inequality in the equation (9), we get:
\begin{IEEEeqnarray*}{l}
\mathtt{g}(Y_{11},X_{11}) + \mathtt{g}(Y_{11},X_{12}) + \cdots + \mathtt{g}(Y_{11},X_{1m}) + m\big\{\mathtt{g}(Y_{22},X_{21}) + \cdots + \mathtt{g}(Y_{mm},X_{m1})\big\} \leq m(2m-1)d
\end{IEEEeqnarray*}
From Set II, we know that $\mathtt{g}(Y_{11},X_{11}) + \mathtt{g}(Y_{11},X_{12}) + \cdots + \mathtt{g}(Y_{11},X_{1m}) \geq (2m-1)d$. Substituting this in the above equation we get:
\begin{IEEEeqnarray}{l}
m\big\{\mathtt{g}(Y_{22},X_{21}) + \mathtt{g}(Y_{33},X_{31}) + \cdots + \mathtt{g}(Y_{mm},X_{m1})\big\} \leq  m(2m-1)d - (2m-1)d \label{ultra}
\end{IEEEeqnarray}
Note that the term $\mathtt{g}(Y_{11},X_{11})$ has been eliminated in the equation (9). In the similar manner as above, we can show that 
\begin{IEEEeqnarray}{l}
m\big\{\mathtt{g}(Y_{22},X_{2j}) + \mathtt{g}(Y_{33},X_{31}) + \cdots + \mathtt{g}(Y_{mm},X_{m1})\big\}
\leq m(2m-1)d - (2m-1)d \text{ for } 2\leq j\leq m
\end{IEEEeqnarray}
Summing up the above $m-1$ inequalities and the inequality in the equation (\ref{ultra}), we get:
\begin{IEEEeqnarray}{l}
m\mathtt{g}(Y_{22},X_{21}) + m\mathtt{g}(Y_{22},X_{22}) + \cdots + m\mathtt{g}(Y_{22},X_{2m}) + m^2\mathtt{g}(Y_{33},X_{31}) + \cdots + m^2\mathtt{g}(Y_{mm},X_{m1})\IEEEnonumber\\
\hfill \leq m^2(2m-1)d - m(2m-1)d \label{here}
\end{IEEEeqnarray}
From Set II, we have $\mathtt{g}(Y_{22},X_{21}) + \mathtt{g}(Y_{22},X_{22}) + \cdots + \mathtt{g}(Y_{11},X_{2m}) \geq (2m-1)d$. Using this inequality in the equation (\ref{here}), we have: 
\begin{IEEEeqnarray}{l}
m^2\mathtt{g}(Y_{33},X_{31}) + \cdots + m^2\mathtt{g}(Y_{mm},X_{m1}) \leq m^2(2m-1)d - 2m(2m-1)d \label{here1}
\end{IEEEeqnarray}

Note that, in the above inequality, the term $\mathtt{g}(Y_{22},X_{21})$ from the equation (10) has been eliminated and thereby the terms $\mathtt{g}(Y_{11},X_{11})$ and $\mathtt{g}(Y_{22},X_{21})$ from the equation (9) have been eliminated.  In this way, eliminating term after term from the left hand side of the equation (9), we get 
\begin{IEEEeqnarray}{l}
m^{m-1}\mathtt{g}(Y_{mm},X_{m1}) \leq m^{m-1}(2m-1)d - (m-1)m^{m-2}(2m-1)d\IEEEeqnarraynumspace\\[6pt]
\text{And hence, }\quad \mathtt{g}(Y_{mm},X_{m1}) \leq \frac{(2m-1)d}{m}
\end{IEEEeqnarray}

Similarly it can be shown that
\begin{IEEEeqnarray}{l}
\mathtt{g}(Y_{mm},X_{mj}) \leq \frac{(2m-1)d}{m} \text{ for } 2\leq j\leq m
\end{IEEEeqnarray}
From Set II, we have that $\mathtt{g}(Y_{mm},X_{m1}) + \mathtt{g}(Y_{mm},X_{m2}) + \cdots + \mathtt{g}(Y_{mm},X_{mm}) \geq (2m-1)d$. Hence, it must be that
\begin{IEEEeqnarray}{l}
\mathtt{g}(Y_{mm},X_{m1}) = \frac{(2m-1)d}{m} 
\end{IEEEeqnarray}

Note that $\gcd(2m-1,m) = 1$. Also, by definition, the rank function is integer valued. Therefore, for $\mathtt{g}(Y_{ii},X_{ij})$ to be a positive integer, $d$ has to be a positive integer multiple of $m$. Thus, by Theorem~\ref{thm:sundar},  for $\mathcal{N}$ to be vector linearly solvable, it is necessary that the message dimension is a positive integer multiple of $m$.

Now we describe a coding scheme for the network achieving an $m$ dimensional vector linear solution. In fact, our coding scheme is a routing scheme. Let the $k^{th}$ symbol of the source $s_{ij}$ is denoted by $X_{ijk}$ where $1\leq k \leq m$. The edge $e_{ii}$ for $1\leq i\leq m$ carries the following $m$ length vector: $[X_{i11},X_{i21},\ldots,X_{im1}]$. And the edge $e_{ij}$ for $1\leq i\leq m$ and $m+1\leq j\leq 2m-1$, carries the vector $[X_{i1(j-m+1)},X_{i2(j-m+1)},\ldots,X_{im(j-m+1)}]$. Now, for any terminal it can be seen that the demands can be satisfied just by routing the required symbols from $v_i$ for $1\leq i \leq 2m-1$ to the terminals. For example, the demands of the terminal $t_1$ is met in the following way: the terminal $t_1$ gets $X_{i11}$ from the message coming from $(v_i,t_1)$ for $1\leq i \leq m$; and $[X_{11(j+1)},X_{21(j+1)},\ldots,X_{m1(j+1)}]$ for $1\leq j\leq m-1$ from the edge $(v_{m+j}, t_1)$. \hfill $\blacksquare$

\section{Conclusion} \label{sec4}
In this paper, we have shown that for any integer $m\geq 2$, there exists a network which has a $m$ dimensional vector linear solution, but has no vector linear solution when message dimension is less than $m$. Our future research objective is to investigate the necessary and sufficient conditions for a network to have a 2-dimensional vector linear solution, but no scalar linear solution.

\end{document}